\documentclass[12pt,a4paper]{article}

\usepackage{graphicx}
\usepackage{amsmath}
\usepackage{bbm}
\usepackage{cancel}
\usepackage[normalem]{ulem}
\usepackage{cite}

\usepackage{hyperref}

\usepackage{amsfonts}
\usepackage{amssymb}
\usepackage{amsmath}
\usepackage{latexsym}
\usepackage{xcolor}
\usepackage{url}
\usepackage[amsmath,thmmarks]{ntheorem}
\usepackage{braket}

\usepackage{tikz}
\usepackage{pgfplots}
\usetikzlibrary{calc}

\usepackage{amssymb}
\usepackage[normalem]{ulem}
\usepackage[mathscr]{eucal}
\usepackage{url}
\usepackage{xcolor}
\usepackage{cite}
\usepackage{graphicx}
\usepackage{psfrag}
\usepackage{subfigure}
\usepackage{latexsym}
\usepackage{amsmath}
\usepackage{enumitem}

\DeclareFontFamily{OT1}{rsfs}{}
\DeclareFontShape{OT1}{rsfs}{CGNPm}{n}{ <-7> rsfs5 <7-10> rsfs7 <10-> rsfs10}{}
\DeclareMathAlphabet{\mycal}{OT1}{rsfs}{CGNPm}{n}

{\catcode `\@=11 \global\let\AddToReset=\@addtoreset}
\AddToReset{equation}{section}

{\catcode `\@=11 \global\let\AddToReset=\@addtoreset}
\AddToReset{figure}{section}

{\catcode `\@=11 \global\let\AddToReset=\@addtoreset}
\AddToReset{table}{section}

\hypersetup{colorlinks=true,linkcolor=black,urlcolor=black,citecolor=black}

\usepackage{accents}

\newcommand\ben{\begin{enumerate}}
\newcommand\een{\end{enumerate}}
\newcommand\bit{\begin{itemize}}
\newcommand\eit{\end{itemize}}

\newcommand{\red}[1]{\tc{red}{#1}}

\newcounter{mnotecount}[section]

\renewcommand{\themnotecount}{\thesection.\arabic{mnotecount}}

\newcommand{\mnote}[1]
{\protect{\stepcounter{mnotecount}}$^{\mbox{\footnotesize
$
\bullet$\themnotecount}}$ \marginpar{
\raggedright\tiny\em
$\!\!\!\!\!\!\,\bullet$\themnotecount: #1} }

\newcommand{\uHS}{\H}

\theoremheaderfont{\hspace{0pt}\sc}
\theorempreskip{14pt}
\theorempostskip{14pt}
\theoremseparator{.}
\theorembodyfont{\itshape}

\newtheorem{theorem}{Theorem}[section]
\newtheorem*{acknowledgement}[theorem]{Acknowledgement}
\newtheorem{lemma}[theorem]{Lemma}
\newtheorem{proposition}[theorem]{Proposition}

\theoremheaderfont{\hspace{0pt}\sc}
\theorempreskip{10pt}
\theorempostskip{10pt}
\theorembodyfont{\upshape}
\newtheorem{remark}[theorem]{Remark}
\newtheorem{example}[theorem]{Example}

\theoremheaderfont{\hspace{0pt}\sc}
\theorempreskip{12pt}
\theorempostskip{12pt}
\theoremsymbol{\ensuremath{\Box}}
\theorembodyfont{\upshape}
\newtheorem*{proof}{Proof}

\newcommand{\jlcax}[1]{}
\newcommand{\eean}{\nonumber\end{eqnarray}}

\newcommand{\kk}[1]{}

\newcommand{\beq}{\begin{equation}}

\newcommand{\FS}       
                  {F}

\newcommand{\HS} 
       {H_{\mbox{\scriptsize volume}}}

{\ptc{this should be removed in the oberwolfach version}}%

\newcommand{\eeal}[1]{\label{#1}\end{eqnarray}}
\newcommand{\bed}{\begin{deqarr}}
\newcommand{\eed}{\end{deqarr}}
\newcommand{\bedl}[1]{\begin{deqarr}\label{#1}}
\newcommand{\eedl}[2]{\arrlabel{#1}\label{#2}\end{deqarr}}

\newcommand{\bel}[1]{\begin{equation}\label{#1}}
\newcommand{\bea}{\begin{eqnarray}}
\newcommand{\bean}{\begin{eqnarray}\nonumber}
\newcommand{\beal}[1]{\begin{eqnarray}\label{#1}}
\newcommand{\eea}{\end{eqnarray}}

\newcommand{\be}{\begin{equation}}
\newcommand{\eeq}{\end{equation}}
\newcommand{\ee}{\end{equation}}
\newcommand{\beqa}{\begin{eqnarray}}
\newcommand{\eeqa}{\end{eqnarray}}
\newcommand{\beqan}{\begin{eqnarray*}}
\newcommand{\eeqan}{\end{eqnarray*}}
\newcommand{\ba}{\begin{array}}
\newcommand{\ea}{\end{array}}

\newcommand{\dx}{\,dx}

\newcommand{\warn}[1]
{\protect{\stepcounter{mnotecount}}$^{\mbox{\footnotesize
$
\bullet$\themnotecount}}$ \marginpar{
\raggedright\tiny\em
$\!\!\!\!\!\!\,\bullet$\themnotecount: {\bf Warning:} #1} }

\newcommand{\R}{\mathbb R}

\newcommand{\ptc}[1]{\mnote{{\bf ptc:}#1}}

\newcommand{\beqar}{\begin{deqarr}}
\newcommand{\eeqar}{\end{deqarr}}

\newcommand{\beaa}{\begin{eqnarray*}}
\newcommand{\eeaa}{\end{eqnarray*}}

\newcommand{\step}{\vskip 3mm}

\newcommand{\one}{\mathbbm{1}}
\newcommand{\p}{\partial}

\newcommand{\tint}{\textstyle\int}

\DeclareMathOperator{\Linop}{\tc{blue}{\tilde{P}}}

\newcommand{\Gfull}{\smash{\tilde{\Gintro}}}
\DeclareMathOperator{\Qop}{Q}
\newcommand{\chargesL}{\Qop_{\textnormal{KID}}}
\newcommand{\chargestr}{\accentset{\circ}{Q}}

\DeclareMathOperator{\Gdiv}{\tc{blue}{B}}
\DeclareMathOperator{\Gdivtilde}{\tilde B}
\DeclareMathOperator{\Gdivdiv}{S}
\newcommand{\dom}{\Omega}
\newcommand{\Gtr}{\smash{\accentset{\circ}{\Gdivdiv}}}

\newcommand{\norm}{\tc{magenta}{W}}

\newcommand{\fdiv}{\varphi}

\newcommand{\rhs}{\tc{blue}{f}}

\DeclareMathOperator{\ran}{ran}
\DeclareMathOperator{\nonlin}{N}

\DeclareMathOperator{\Ricci}{R}

\newcommand{\pqed}{\hfill$\Box$}

\newcommand{\noA}{\tc{blue}{A}}

\newcommand{\oA}{\accentset{\circ}{\noA}}
\newcommand{\Atr}{\tau}

\DeclareMathOperator{\tr}{tr}
\DeclareMathOperator{\RicciScalar}{R}
\newcommand{\bb}{\tc{orange}{b}}
\newcommand{\pert}{\tc{orange}{h}}
\newcommand{\newpert}{\tc{orange}{A}}

\DeclareMathOperator{\Gintro}{\tc{orange}{G}}
\DeclareMathOperator{\linscalar}{\tc{blue}{P}}
\DeclareMathOperator{\Piscalar}{\tc{blue}{\Pi}}
\newcommand{\fintro}{\tc{blue}{\varphi}}

\newcommand{\tc}[2]{#2}

\renewcommand{\R}{\mathbbm{R}}
\renewcommand{\H}{\mathbbm{H}}
\renewcommand{\subset}{\subseteq}

\let\oldabstract\abstract
\let\oldendabstract\endabstract
\makeatletter
\renewenvironment{abstract}
{%
               {\list{}{\setlength{\leftmargin}{7.9mm} 
                        \listparindent 1.5em%
                        \itemindent    \listparindent%
                        \rightmargin   \leftmargin%
                        \parsep        \z@ \@plus\p@}%
                \item\relax}%
               {\endlist}%
\oldabstract}
{\oldendabstract}
\makeatother

\begin{document}

\title{A Bogovski\v i-type operator \\ for Corvino-Schoen hyperbolic gluing}
\date{}
\author{%
Piotr T. Chru\'{s}ciel, 
Albachiara Cogo, 
Andrea N\"utzi}
\maketitle

\begin{abstract}
We construct a solution operator for the
linearized constant scalar curvature equation at hyperbolic space
in space dimension larger than or equal to  two.
The solution operator has good support propagation properties
and gains two derivatives relative to standard norms.
It can be used for Corvino-Schoen-type hyperbolic gluing,
partly extending the recently introduced 
Mao-Oh-Tao gluing method to the hyperbolic setting.
%
\end{abstract}

\tableofcontents


\section{Introduction}\label{s23XI16.1}

Mao, Oh and Tao \cite{MaoOhTao}
have shown how to simplify the gluing constructions of Corvino and Schoen~\cite{Corvino,CorvinoSchoen2}
for three-dimensional asymptotically flat general relativistic vacuum initial data sets. 
The simplification comes from using Bogovski\v i-type solution 
operators for the linearized equations,
which preserve the support of the sources and have good differentiability properties.
Similar methods \cite{MaoTao} led to simplifications and improvements of the gluing constructions
of Carlotto and Schoen~\cite{CarlottoSchoen} in dimensions larger than or equal to three, in particular
they provide better asymptotics of the glued solutions.  

 The aim of this note is to point out the existence of a Bogovski\v i-type 
solution operator for the 
linearization of the constant scalar curvature equation  
at hyperbolic space, 
with the same good support and differentiability properties. 
It thus can be used in nonlinear problems, 
e.g.~in Corvino-type gluing near hyperbolic space,
and in particular to provide simpler proofs of some of
the small-data gluing results of \cite{ChDelayAH}. 
See Section \ref{s13VIII24} for details. 
\step

The constant scalar curvature equation for an $n$-dimensional
asymptotically hyperbolic Riemannian manifold $(M,g)$ is
\begin{equation}\label{eq:scalcurveq}
\RicciScalar(g) = -n(n-1)
 \,,
\end{equation}
where $\RicciScalar(g)$ is the scalar curvature of the metric $g$.
Recall that \eqref{eq:scalcurveq} is equivalent to the time-symmetric 
general relativistic constraint equations with cosmological constant
$2\Lambda=-n(n-1)$.

Here we consider the case where the manifold is the upper half-space
$$M=\uHS^n=\{x=(x^1,\dots,x^n)\mid x^1>0\}\,.$$
Then a trivial solution of  \eqref{eq:scalcurveq} is given by the hyperbolic metric 
$$\bb=\frac{1}{(x^1)^2} (dx^1\otimes dx^1 +\dots + dx^n\otimes dx^n)\,.$$

We consider the linearization of \eqref{eq:scalcurveq} at $\bb$.
This is given by the under-determined elliptic linear operator 
$\linscalar:C^\infty(\uHS^n,S^2(T^*\uHS^n)) \to C^\infty(\uHS^n,\R)$,
\begin{align} \label{eq:Pdef}
\linscalar(\pert) 
= D^i D^j \pert_{ij} -D^i D_i \pert^{j}{}_{j} - \Ricci^{ij} \pert_{ij}
\,, 
\end{align}
where  $S^2(T^*\uHS^n)$ denotes the bundle of two-covariant symmetric tensors,
$D_i$ and $\Ricci_{ij}$ are the covariant derivative,
respectively the Ricci curvature, of $\bb$,
and where we use the index conventions detailed in Remark \ref{rem:conventions}.

Our main proposition provides a solution operator for the linear equation
\begin{equation}\label{eq:lin}
\linscalar(\pert) = \rhs
\end{equation}
with good support and differentiability properties.
We discuss necessary integrability conditions of this equation
when the source $\rhs$ has compact support in $\uHS^n$,
and one seeks a solution $\pert$ with compact support.
%
Recall that the space of static KIDs for the hyperbolic metric
is given by the kernel of the formal adjoint of $\linscalar$.
Explicitly, this is spanned over $\R$ by the $n+1$ functions 
\begin{equation}\label{9VII24.51}
 {\frac{1+\red{|x|}^2}{2x^1}\,,}
\
 {\frac{1-\red{|x|}^2}{2x^1}\,,}
\ 
\frac{x^2}{x^1}\,,
\ \ldots\,,
\frac{x^{n}}{x^1}\,,
\end{equation}
which we will denote by $\kappa_1,\dots,\kappa_{n+1}$.
For all $\pert$ and all $a=1\dots n+1$ the product $\kappa_a \linscalar(\pert)$
is a divergence, in fact the identity 
\begin{equation}\label{eq:kappaPpert}
\kappa_a \linscalar(\pert) = D^i V^{(a)}_i(\pert)
\end{equation} 
holds where by definition 
\begin{equation*}
V^{(a)}_i(\pert)
= 
\kappa_a (D^{j}\pert_{ij})
-
(D^j \kappa_a) \pert_{ij}
-
\kappa_a( D_i \pert^j{}_j)
+
(D_i\kappa_a) \pert^j{}_j \,.
\end{equation*}
Therefore for all $\pert$ with compact support,
\begin{equation}
\label{eq:kappaP}
\int_{\uHS^n} \kappa_a \linscalar(\pert)\, d\mu_{\bb} = 0
\end{equation}
with $d\mu_{\bb}=\frac{1}{(x^1)^n}dx^1\cdots dx^n$.
These are $n+1$ necessary
integrability conditions for \eqref{eq:lin},
in the compact support case.
We now state the main proposition, which shows in
particular that these integrability conditions are sufficient for existence of a solution with compact support.

For an open subset $\Omega\subset\uHS^n$ we denote by 
$C^\infty_c(\Omega,\R)$ the space of smooth functions with compact
support in $\Omega$, analogously for sections of $S^2 (T^*\dom)$.

\begin{proposition}\label{prop:main}
Let $\dom\subset\uHS^n$ be a 
bounded connected open set with smooth boundary such that
$\bar\dom\subset\uHS^n$. Then there exists an $\R$-linear operator
\[ 
\Gintro:\; C^\infty_c(\dom,\R)\to C^\infty_c(\dom,S^2 (T^*\dom))
\]
with the following properties:

\begin{enumerate}[label=\textnormal{(a\arabic*)}]
\item \label{item:introeq} 
The map 
$\Piscalar:C^\infty_c(\dom,\R)\to C^\infty_c(\dom,\R)$ 
defined by 
%
\[ 
\linscalar \Gintro = \one -\Piscalar
\]
has rank $n+1$,  is a projection ($\Piscalar^2=\Piscalar$),
and is explicitly given as follows:
Let $\kappa_1,\dots,\kappa_{n+1}$ be the static KIDs \eqref{9VII24.51}.
There exist functions $\fintro_1,\dots,\fintro_{n+1}\in C^\infty_c(\Omega,\R)$
such that for all $\rhs\in C^\infty_c(\dom,\R)$,
\begin{equation}
\label{eq:Pidef}
\Piscalar(\rhs) \;=\; 
\sum_{a=1}^{n+1}
\fintro_a \int_{\dom} \kappa_a\, \rhs \,d\mu_{\bb} \,.
\end{equation}
\item\label{item:introest}
For all $s\in\R$ and $1<p<\infty$ there exists a constant $C>0$ such that for all 
$i=1\dots n$  and $\rhs\in C^\infty_c(\Omega,\R)$,
\begin{align*}
\|\Gintro(\rhs)\|_{\norm^{s,p}} &\;\le\; C \|\rhs\|_{\norm^{s-2,p}}\,, \\
\|[\Gintro,\p_{x^i}](\rhs)\|_{\norm^{s,p}} &\;\le\; C \|\rhs\|_{\norm^{s-2,p}}\,,
\end{align*}
where $\|\cdot\|_{\norm^{s,p}}$ are the standard $L^p$-based Sobolev norms.
Analogous estimates hold for standard H\"older norms with exponent $0<\alpha<1$.
\end{enumerate}
\end{proposition}

The proof appears at the end of Section \ref{s11VII24.1}.

\begin{acknowledgement}
{\rm
Many useful discussions with Bobby Beig and Erwann Delay are acknowledged.
We are thankful to the Erwinn Schr\"odinger Institute in Vienna for hospitality during an initial stage of this work.
AC was supported by the DAAD 'Forschungsstipendien 
f\"ur Doktorandinnen und Doktoranden' during her research stay in Vienna.
AN is supported by the Swiss National Science Foundation, project number P500PT-214470. PTC's research was supported in part by the NSF under Grant No. DMS-1928930 while the author was in residence at the Simons Laufer Mathematical Sciences Institute (formerly MSRI) in Berkeley during the Fall 2024 semester. 
}
\end{acknowledgement}


\section{Linearized scalar curvature equation}
In this section we derive a simple formula for the 
linear operator $\linscalar$ in \eqref{eq:Pdef}.
We will use the conventions detailed in the following remark.
\begin{remark}\label{rem:conventions}
For a symmetric two-tensor $\pert\in C^\infty(\uHS^n,S^2(T^*\uHS^n))$
we denote by $\pert_{ij}=\pert(\p_{x^i},\p_{x^j})$
its components relative to the coordinates $x$.
Indices are lowered with $\bb_{ij}=\bb(\p_{x^i},\p_{x^j})$
and raised with its inverse $\bb^{ij}$.
All repeated indices are summed over $1\dots n$,
regardless of their position.
We denote by $\delta_{ij}=\delta^{ij}=\delta^i_j$ the Kronecker delta.
\pqed
\end{remark}

Let $\pert\in C^\infty(\uHS^n,S^2(T^*\uHS^n))$.
Using the fact that $\Ricci_{ij}=(1-n)\bb_{ij}$ we have 
\[ 
\linscalar(\pert) = D^i D^j (  \pert_{ij} -   \pert^\ell{}_\ell  \bb_{ij} ) 
- (1-n) \pert^{\ell}{}_{\ell} \,.
\]
Thus if we define $\smash{\tilde{\pert}}=\pert -   \pert^\ell{}_\ell\,  \bb$
then we obtain 
\[ 
\linscalar(\pert) = D^i D^j \smash{\tilde{\pert}}_{ij} - \smash{\tilde{\pert}}^i{}_i\,.
\]
The non-vanishing Christoffel symbols of the hyperbolic metric are
\begin{align*}
\Gamma^1_{11} =  -\tfrac{1 }{x^1}
\,,
\quad
\Gamma^1_{pq} = \tfrac{1}{x^1}\delta_{pq}
  \,,
\quad
\Gamma^p_{1q}= -\tfrac{1}{x^1} \delta^p_q\,,
\end{align*}
where $p,q=2\dots n$. Then by direct calculation,
\begin{align*}
  D^iD^j \smash{\tilde{\pert}}_{ij}
    &=
   (x^1)^{n+1} \partial_i \partial_j
   \left( \tfrac{1}{ (x^1)^{n+1}}\smash{\tilde{\pert}}^{ij}
   \right)
   +(x^1)^n \partial_1 \left(\tfrac{1}{ (x^1)^{n+1}} \delta_{ij}\smash{\tilde{\pert}}^{ij}\right)
    \,,
\end{align*}
where $\p_i = \p_{x^i}$.
Thus, using $\smash{\tilde{\pert}}^i{}_i= \frac{1}{(x^1)^2}\delta_{ij}\smash{\tilde{\pert}}^{ij}$,
\[ 
\linscalar(\pert) = 
(x^1)^{n+1} \partial_i \partial_j
   \left( \tfrac{1}{(x^1)^{n+1}}\smash{\tilde{\pert}}^{ij}
   \right)
   +(x^1)^{n+1} \partial_1 \left(\tfrac{1}{ (x^1)^{n+2}} (\delta_{ij}\smash{\tilde{\pert}}^{ij})\right)\,.
\]
We have therefore proved the following lemma:

\begin{lemma}\label{lem:Pnice}
Let $\pert\in C^\infty(\uHS^n,S^2(T^*\uHS^n))$. Define
\begin{equation}\label{eq:changevar}
\newpert = \tfrac{1}{(x^1)^{n+1}}(\pert - \pert^i{}_{i}\,\bb)\,.
\end{equation}
Then 
$\pert = (x^1)^{n+1} (\newpert + \frac{1}{1-n} \newpert^i{}_{i}\,\bb)$
and 
\begin{equation*}
\linscalar(\pert) 
= 
(x^1)^{n+1} 
\left(\p_{i}\p_{j} \newpert^{ij} 
+  \partial_1 (\tfrac{1}{x^1} \newpert^{ii})
\right)\,.
\end{equation*}
\end{lemma}


\section{Construction of the solution operator}

In this section we construct the operator $\Gintro$ of Proposition \ref{prop:main}.
The main tool is a right inverse for the divergence operator
on Euclidean space originally introduced by Bogovski\v i \cite{Bogovskii1,Bogovskii2}.
We recall it in the next lemma.

\begin{lemma}\label{lem:bogdiv}
Let $\Omega\subset\uHS^n$ be as in Proposition \ref{prop:main}.
There exists an $\R$-linear map 
$\Gdiv:C^\infty_c(\Omega,\R)\to C^\infty_c(\Omega,\R^{n})$
with the following properties:
\begin{enumerate}[label=\textnormal{(b\arabic*)}]
\item \label{item:gdivid}
There exists a function $\fdiv\in C^\infty_c(\Omega,\R)$ such that 
$\int_\Omega \fdiv \dx=1$ where $\dx$ is the Lebesgue measure,
and such that for all $\rhs\in C^\infty_c(\Omega,\R)$,
\begin{equation}\label{eq:divG}
\p_{i} \Gdiv^i (\rhs) \;=\; \rhs- \fdiv \int_\Omega \rhs \dx \,.
\end{equation}
\item \label{item:gdivest}
For all $s\in\R$ and $1<p<\infty$ 
there exists a constant $C>0$ such that for all 
$i=1\dots n$ and 
$\rhs\in C^\infty_c(\Omega,\R)$,
\begin{align*}
\|\Gdiv(\rhs)\|_{\norm^{s,p}} &\;\le\; C \|\rhs\|_{\norm^{s-1,p}}\,,\\
\|[\Gdiv,\p_{i}](\rhs)\|_{\norm^{s,p}} &\;\le\; C \|\rhs\|_{\norm^{s-1,p}}\,.
\end{align*}
Analogous estimates hold for H\"older norms with exponent $0<\alpha<1$.

\item \label{item:integralsGdiv}
For all $i=1\dots n$ and $\rhs\in C^\infty_c(\Omega,\R)$,
\begin{align}
\textstyle
\int_{\Omega} \Gdiv^i(\rhs) \dx
&=
\textstyle
-\int_{\Omega} x^i \rhs \dx + m^i \int_{\Omega} \rhs \dx
 \,,
\label{eq:intg1}  \\
\textstyle \int_{\Omega}x^j \Gdiv^j(\rhs) \dx
&=
\textstyle -\tfrac12 \int_{\Omega} |x|^2 \rhs \dx + \tfrac12m^{jj}\int_{\Omega} \rhs \dx
\,,
\label{eq:intg2}
\end{align}
where by definition $m^i=\int_{\Omega} x^i \fdiv \dx$
and $m^{ij}=\int_{\Omega} x^ix^j \fdiv \dx$,
and where we use 
the summation convention in Remark \ref{rem:conventions}.
\end{enumerate}
\end{lemma}
\begin{proof}
Such an operator exists by \cite[Section 4.1]{OhTataru}
(note that $\Omega$ can be covered by finitely 
many open sets that are star-shaped with respect to a ball),
or alternatively by \cite[Theorem 1]{Andrea}.
More precisely, the references imply that there exists an operator
$\Gdivtilde:C^\infty_c(\R^n,\R)\to C^\infty_c(\R^n,\R^{n})$
that satisfies \ref{item:gdivid},
$\Gdivtilde (C^\infty_c(\Omega,\R))\subset C^\infty_c(\Omega,\R^n)$,
and such that if $\chi\in C^\infty_c(\R^{n},\R)$
then $\Gdivtilde\chi$ is a pseudodifferential operator of order $-1$ on $\R^{n}$.
Define $\Gdiv = \Gdivtilde|_{C^\infty_c(\Omega,\R)}$.
Then $\Gdiv$ satisfies \ref{item:gdivid},
and \ref{item:gdivest} using standard mapping properties for
pseudodifferential operators,
see \cite[Chapter 13.6, Proposition 6.5]{Taylor96} for Sobolev norms,
and \cite[Chapter 13.8, Proposition 8.5]{Taylor96} for H\"older norms.

We show \ref{item:integralsGdiv}.
For \eqref{eq:intg1} multiply \eqref{eq:divG} with $x^i$ and integrate, which yields 
\[
\textstyle
\int_{\dom} x^i \p_{j} \Gdiv^j(\rhs) \dx = \int_{\dom} x^i \rhs \dx - 
m^i \int_{\dom} \rhs\dx\,.
\]
Integrating by parts yields \eqref{eq:intg1}.
For \eqref{eq:intg2} multiply \eqref{eq:divG} with $\red{|x|}^2$
and proceed analogously.
\end{proof}

Given such a right inverse $\Gdiv$ for the divergence operator, 
we obtain the solution operator $\Gintro$ for $\linscalar(\pert)=\rhs$ in two steps.
We first use $\Gdiv$ to obtain a right inverse 
for the double divergence on symmetric trace-free two-tensors (Section \ref{sec:tracefree}).
Using this operator one can then solve separately
for the trace of $\pert$ and the trace-free part of $\pert$ (Section \ref{s11VII24.1}).
In particular, given $\Gdiv$ the construction of $\Gintro$ 
requires no further analysis.


\subsection{Trace-free symmetric double divergence}\label{sec:tracefree}

Let $C^\infty_c(\dom,S^2\R^n)$ be the space
of symmetric $n\times n$ matrix valued functions with 
compact support in $\dom$.
In this section we use the right inverse $\Gdiv$ of the divergence
from Lemma \ref{lem:bogdiv} to obtain a solution operator for the equation
\begin{equation}\label{eq:tracelesseq}
\p_{i}\p_{j}\noA^{ij} \;=\; \rhs\,,
\end{equation}
where the source $\rhs$ has compact support in $\dom$
and where we seek $\noA\in C^\infty_c(\dom,S^2\R^n)$
with vanishing matrix trace, $\tr(\noA)=\noA^{ii}=0$.
Define
\begin{equation*}
\chargestr(\rhs)
=
\int_\dom
\Big(
1\,,\,x^1\,,\,x^2\,,\,\dots,x^n\,,\,\red{|x|}^2
\Big) \rhs\,\dx
\;\in\; \R^{n+2}
\,.
\end{equation*}
Integrating by parts one obtains that for all 
$\noA\in C^\infty_c(\dom,S^2\R^n)$ with $\noA^{ii}=0$,
\begin{equation}\label{eq:Qtrint}
\smash{\chargestr}(\p_{i}\p_{j}\noA^{ij}) = 0\,.
\end{equation}
Thus $\smash{\chargestr}(\rhs)=0$ are $n+2$ necessary integrability conditions
for \eqref{eq:tracelesseq}.

\step
As a first step we need the next lemma,
this is similar to \cite[Lemma 2.2]{MaoOhTao}.
\begin{lemma}\label{lem:divdiv}
Let $\dom$, $\Gdiv$, $\fdiv$, $m^i$, $m^{ij}$ be as in Lemma \ref{lem:bogdiv}.
Define the operator
$
\Gdivdiv: C^\infty_c(\dom,\R) \to C^\infty_c(\dom,S^2\R^{n})
$ 
by 
$$\Gdivdiv^{ij}(\rhs) = 
\tfrac12\Gdiv^i(\Gdiv^j(\rhs))+ \tfrac12\Gdiv^j(\Gdiv^i(\rhs))\,.$$
Then:
\begin{enumerate}[label=\textnormal{(c\arabic*)}]
\item \label{item:divdiveq}
For all $\rhs\in C^\infty_c(\dom,\R)$,
\[\textstyle
\p_{i}\p_{j}\Gdivdiv^{ij}(\rhs)
=
\rhs
-\left(\left(\fdiv + m^j \p_{j}\fdiv \right)\int_{\dom} \rhs\dx
- \p_{j}\fdiv \int_{\dom} x^j \rhs\dx\right)\,.
\]

\item \label{item:divdivest}
The estimates in point \ref{item:introest} of Proposition~\ref{prop:main} hold for $\Gdivdiv$.

\item \label{item:divdivtr}
For all $\rhs\in C^\infty_c(\dom,\R)$,
\begin{align*}
\textstyle
\int_{\dom} \Gdivdiv^{ii}(\rhs) \dx
&=
\textstyle
(m^i m^i -\tfrac12  m^{ii}) \int_{\dom} \rhs \dx
-m^i  \int_{\dom} x^i \rhs \dx
+\tfrac12 \int_{\dom} \red{|x|}^2 \rhs \dx \,.
\end{align*}
\end{enumerate}
\end{lemma}

\begin{proof}
\ref{item:divdiveq}:
Using point~\ref{item:gdivid} of Lemma~\ref{lem:bogdiv} we have
\begin{align*}
\p_{i}\p_{j}\Gdivdiv^{ij}(\rhs)
=
\p_{j}\p_{i}\Gdiv^i(\Gdiv^j(\rhs))
&=
\textstyle
\p_{j}\Gdiv^j(\rhs) - \p_{j}\fdiv \int_{\dom} \Gdiv^j(\rhs) \dx \\
&=
\textstyle
\rhs- \fdiv\int_{\dom} \rhs \dx  - \p_{j}\fdiv \int_{\dom} \Gdiv^j(\rhs) \dx
 \,.
\end{align*}
Together with \ref{item:integralsGdiv} one obtains \ref{item:divdiveq}.

\ref{item:divdivest}: This follows directly from \ref{item:gdivest},
for the commutator estimate also use
$[\Gdiv^k\Gdiv^\ell,\p_{i}]=\Gdiv^k[\Gdiv^\ell,\p_{i}]+[\Gdiv^k,\p_{i}]\Gdiv^\ell$
for all $k,\ell,i=1\dots n$.

\ref{item:divdivtr}: This follows directly from \ref{item:integralsGdiv}.
\end{proof}

Suppose that $\rhs\in C^\infty_c(\dom,\R)$ satisfies the
integrability conditions $\smash{\chargestr}(\rhs)=0$ from \eqref{eq:Qtrint}.
Then by \ref{item:divdiveq}, respectively \ref{item:divdivtr},
\begin{align*}
\p_{i}\p_{j}\Gdivdiv^{ij}(\rhs)=\rhs\,,
\qquad
\textstyle\int_\dom \Gdivdiv^{ii}(\rhs) \dx=0\,.
\end{align*}
The second identity and point \ref{item:gdivid} of Lemma~\ref{lem:bogdiv} imply
$\Gdivdiv^{ii}(\rhs)=\p_{j}\Gdiv^j(\Gdivdiv^{ii}(\rhs))$.
This motivates making the following ansatz
for a trace-free solution of \eqref{eq:tracelesseq}:
\begin{equation}\label{eq:ansatz}
\Gdivdiv^{ij}(\rhs) + (X^\ell{}_{k})^{ij} \p_{\ell} \Gdiv^{k}(T)
\quad\text{with}\quad T=\Gdivdiv^{ii}(\rhs)\,,
\end{equation}
where $(X^\ell{}_{k})^{ij}\in \R$
is symmetric in $ij$, $(X^\ell{}_{k})^{ij}=(X^\ell{}_{k})^{ji}$.
We require that:
\begin{itemize}
\item
The total symmetrization of $(X^\ell{}_{k})^{ij}$
in $\ell,i,j$ vanishes, for all $k$.
\\
Then
$\p_{i}\p_{j} \big((X^\ell{}_{k})^{ij} \p_{\ell} \Gdiv^{k}(T) \big)=0$,
hence \eqref{eq:ansatz} still solves \eqref{eq:tracelesseq}.
\item
$(X^\ell{}_k)^{ii} = -\delta^\ell_k$ for all $k,\ell$.\\
Then
$(X^\ell{}_{k})^{ii} \p_{\ell} \Gdiv^{k}(T) = - \p_{\ell} \Gdiv^{\ell}(T) = -T$,
hence \eqref{eq:ansatz} is trace-free.
\end{itemize}
The following choice for $(X^\ell{}_k)^{ij}$
satisfies these properties:
%
\begin{equation}\label{eq:choiceX}
(X^\ell{}_k)^{ij} =
\tfrac{1}{n-1}
\left(
  \tfrac 12 (\delta^{\ell i} \delta^j_k +
  			\delta^{\ell j} \delta^i_k)
  - \delta ^{ij}\delta^\ell_k
  \right)
  \,.
\end{equation}
%


\begin{lemma}\label{lem:Gtr}
Let $\dom$, $\fdiv$, $\Gdiv$, $\Gdivdiv$ be as in Lemma \ref{lem:bogdiv}
and Lemma \ref{lem:divdiv},
and let $(X^\ell{}_{k})^{ij}$ be as in \eqref{eq:choiceX}.
Define 
$\Gtr: C^\infty_c(\dom,\R) \to C^\infty_c(\dom,S^2\R^{n})$ by
\begin{align}\label{eq:Gtrformula}
\Gtr^{ij}(\rhs)
&\;=\;
\Gdivdiv^{ij}(\rhs) + (X^\ell{}_{k})^{ij} \p_{\ell} \Gdiv^{k}\big(\Gdivdiv^{pp}(\rhs)\big)\,.
\end{align}
Then:
\begin{enumerate}[label=\textnormal{(d\arabic*)}]
\item \label{item:Gtreq}
For all $\rhs\in C^\infty_c(\dom,\R)$ one has
$\p_{i}\p_{j}\Gtr^{ij}(\rhs)=\p_{i}\p_{j}\Gdivdiv^{ij}(\rhs)$.
\item \label{item:Gtrest}
The estimates of point \ref{item:introest} of Proposition~\ref{prop:main} hold for $\Gtr$.

\item \label{item:Gtretr}
For all $\rhs\in C^\infty_c(\dom,\R)$ one has
$\Gtr^{ii}(\rhs)=\fdiv \int_{\dom}\Gdivdiv^{ii}(\rhs)\dx$.
\end{enumerate}
\end{lemma}

In particular, if $\smash{\chargestr}(\rhs)=0$ then
$\noA=\Gtr(\rhs)$ is trace-free and solves \eqref{eq:tracelesseq}.

\begin{proof}
Denote the second term on the right
of \eqref{eq:Gtrformula} by $X^{ij}(\rhs)$.

\ref{item:Gtreq}:
By $\p_{i}\p_{j}X^{ij}(\rhs)=0$, using our choice of $(X^\ell{}_k)^{ij}$.

\ref{item:Gtrest}: This follows from points \ref{item:gdivest} of Lemma~\ref{lem:bogdiv} 
and \ref{item:divdivest} of Lemma~\ref{lem:divdiv}.

\ref{item:Gtretr}: We have 
\begin{align*}
X^{ii}(\rhs)
&=
-\p_{\ell} \Gdiv^{\ell}(\Gdivdiv^{ii}(\rhs))
=
-\Gdivdiv^{ii}(\rhs)
+\textstyle \fdiv \int_{\dom}\Gdivdiv^{ii}(\rhs)\dx\,,
\end{align*}
where for the last equality we use \ref{item:gdivid}.
This implies \ref{item:Gtretr}.
\end{proof}


\subsection{Linearized constant scalar curvature equation}\label{s11VII24.1}

\newcommand{\vecsp}{\Psi}
\newcommand{\functr}{\tc{blue}{\psi}}

Define the linear operator $\Linop:C^\infty(\Omega,S^2\R^n)\to C^\infty(\Omega,\R)$,
\[
\Linop(\noA) =
\partial_{i} \partial_{j} \noA^{ij}
   + \partial_{1} \big(\tfrac{1}{x^1} \noA^{ii}\big)\,.
\]
By Lemma \ref{lem:Pnice} this agrees with the 
linearized scalar curvature operator $\linscalar$
up to the change of variables \eqref{eq:changevar} and the 
overall factor $(x^1)^{n+1}$.
Here we use the operator $\Gtr$ from Lemma \ref{lem:Gtr}
to obtain a solution operator for
\begin{align}\label{eq:Leq}
\Linop(\noA) = \rhs\,,
\end{align}
where $\rhs$ has compact support in $\dom$
and we seek $\noA\in C^\infty_c(\dom,S^2\R^n)$.
Define 
\begin{equation*}
\chargesL(\rhs)
=
\int_\Omega 
\Big(
1\,,\,x^2\,,\,\dots,x^n\,,\,\red{|x|}^2
\Big)
\rhs\, \dx
\;\in\; \R^{n+1}\,.
\end{equation*}
By \eqref{eq:kappaP}, for all $\noA$ with compact support we have 
\begin{align}\label{eq:cLnoA}
\chargesL(\Linop(\noA)) = 0
\end{align}
and thus $\chargesL(\rhs)=0$ are
$n+1$ necessary integrability conditions for \eqref{eq:Leq}.

We indicate how one can solve \eqref{eq:Leq} when $\chargesL(\rhs)=0$. 
Decompose
$$
\noA = \oA + \tfrac{1}{n} \Atr \one
$$
where $\oA$ is trace-free symmetric and $\Atr$ is a function,
and where $\one$ is the $n\times n$ identity matrix.
Then \eqref{eq:Leq} reads
\begin{equation}
\label{eq:L1def}
\p_{i}\p_{j} \oA^{ij} = \rhs - \Linop\big(\tfrac{1}{n}\Atr\one \big)\,,
\end{equation}
with $\Linop(\tfrac{1}{n}\Atr\one ) 
=  \tfrac{1}{n} \p_{i} \p_{i} \Atr + \p_{1}(\tfrac{1}{x^1}\Atr)$.
Note that for all $\Atr$ with compact support,
\begin{align}\label{eq:AddCharge}
\textstyle\int_\Omega x^1 \Linop\big(\tfrac{1}{n}\Atr\one \big)\dx &= - \textstyle\int_{\Omega} \frac{1}{x^1}\Atr\dx\,.
\end{align}
Thus by choosing $\Atr$ appropriately one can achieve
that $f'=f-\Linop(\tfrac{1}{n}\Atr\one )$ satisfies 
both $\chargesL(f')=0$ and $\int_{\Omega} x^1 f'\dx=0$,
which means that $f'$ satisfies the $n+2$ integrability conditions
$\chargestr(f')=0$ for the trace-free symmetric double divergence
in \eqref{eq:Qtrint}. 
Thus $\oA=\Gtr(\rhs')$
is trace-free and solves \eqref{eq:L1def}. 


This is implemented in the following proposition.

\begin{lemma}\label{lem:linopinv}
Let $\Omega$, $\fdiv$, $\Gtr$ be as in Lemma \ref{lem:Gtr}.
Set $\functr=-x^1\fdiv$.
Define $\Gfull:C^\infty_c(\dom,\R)\to C^\infty_c(\dom,S^2\R^{n})$
by
\[
\Gfull(\rhs) =
\textstyle
\Gtr\left( \rhs -  \Linop\left(\tfrac{1}{n}\functr(\int_{\dom} x^1 \rhs\dx) \one\right)  \right)
+ \frac{1}{n}\functr  (\int_{\dom}x^1\rhs\dx) \one\,.
\]
Then:
\begin{enumerate}[label=\textnormal{(e\arabic*)}]
\item \label{item:Gfulleq}
There exist $\fdiv_1,\dots,\fdiv_{n+1}\in C^\infty_c(\Omega,\R)$
such that for all $\rhs\in C^\infty_c(\Omega,\R)$,
\[
\Linop (\Gfull (\rhs)) = \rhs- \sum_{a=1}^{n+1} \fdiv_a\, \chargesL^a (\rhs) \,.
\]
\item \label{item:Gfullest}
The estimates of point  \ref{item:introest} of Proposition~\ref{prop:main} hold for $\Gfull$.
\end{enumerate}
\end{lemma}

\begin{proof}
\ref{item:Gfulleq}:
Abbreviate $\Atr = \functr \int_{\dom} x^1 \rhs\dx$
and $\rhs'=\rhs - \Linop(\frac{1}{n}\Atr \one)$.
Then 
\begin{align}
\chargesL(\rhs') 
&=\chargesL(\rhs) - \chargesL\big(\Linop\big(\tfrac{1}{n}\Atr\one \big)\big)
= \chargesL(\rhs)\,, 
\label{eq:chargesrho'1}\\
\textstyle\int_{\Omega} x^1 \rhs'\dx 
&= 
\textstyle
\int_{\Omega} x^1 \rhs\dx 
- 
\int_{\Omega} x^1 \Linop\big(\tfrac{1}{n}\Atr\one \big) \dx \nonumber \\
&= 
\textstyle
\big(1+\int_{\Omega} \frac{1}{x^1}\functr\dx\big)\int_{\Omega} x^1 \rhs\dx 
=0 \,.\label{eq:chargesrho'2}
\end{align}
For \eqref{eq:chargesrho'1} we use \eqref{eq:cLnoA},
and for \eqref{eq:chargesrho'2} we use \eqref{eq:AddCharge}
and $\int_{\Omega} \frac{1}{x^1}\functr\dx
=-\int_{\Omega}\fdiv\dx =-1$.
By definition of \smash{$\Linop$} and \smash{$\Gfull$} we have
\begin{align*}
\Linop(\Gfull(\rhs))
&=
\p_i\p_j\Gtr^{ij}(\rhs') + \p_1\big(\tfrac{1}{x^1}\Gtr^{ii}(\rhs')\big) 
+ \Linop\big(\tfrac{1}{n}\Atr \one\big)\\
&=
\p_i\p_j\Gdivdiv^{ij}(\rhs') 
+ \p_1(\tfrac{1}{x^1} \fdiv)\tint_{\dom}\Gdivdiv^{ii}(\rhs')\dx
+ \Linop\big(\tfrac{1}{n}\Atr \one\big)
\end{align*}
where for the second equality we use \ref{item:Gtreq} and \ref{item:Gtretr}.
We now rewrite the first term using \ref{item:divdiveq},
and the second term using \ref{item:divdivtr},
then we simplify using \eqref{eq:chargesrho'1} and \eqref{eq:chargesrho'2}.
This yields \ref{item:Gfulleq} with 
\begin{align*}
\fdiv_a = 
\begin{cases}
 \fdiv + m^i \p_i\fdiv - (m^im^i -\frac12 m^{ii})\p_1\big(\frac{1}{x^1}\fdiv\big) 
 & \text{if $a=1$}\,,\\
 -\p_a\fdiv+ m^a\p_1\big(\frac{1}{x^1}\fdiv\big)
 & \text{if $a=2\dots n$}\,,\\
 -\frac12 \p_1\big(\frac{1}{x^1}\fdiv\big)
 & \text{if $a=n+1$}\,.
\end{cases}
\end{align*}

\ref{item:Gfullest}:
By linearity of $\Gtr$ we can write
\[
\Gfull(\rhs) =
\Gtr(\rhs)
+\Big(\tfrac{1}{n}  \functr \one
- \Gtr\big(\Linop\big(\tfrac{1}{n}\functr \one\big)\big)
\Big) \textstyle\int_{\dom} x^1 \rhs\dx \,.
\]
The first term satisfies the estimates by \ref{item:Gtrest},
for the second they are clear.
\end{proof}

We are now ready to pass to the 

\begin{proof}[of Proposition \ref{prop:main}]
Define $\Gintro$ by 
$$
\Gintro^{ij}(\rhs)
=
(x^1)^{n+1} \left( 
\Gfull^{ij}
+
\tfrac{1}{1-n} \bb^{ij} \bb_{k\ell}\Gfull^{k\ell}
\right)\big(\tfrac{1}{(x^1)^{n+1}}\rhs\big) \,.
$$
We check \ref{item:introeq}: 
 We have \eqref{eq:Pidef} 
by Lemma \ref{lem:Pnice} and \ref{item:Gfulleq};
the map $\Piscalar$ has rank $n+1$ since 
by \eqref{eq:Pidef} the rank is $\le n+1$
and by \eqref{eq:kappaP} the rank is $\ge n+1$;
we have $\Piscalar^2=\Piscalar$ since applying 
$\Piscalar$ to \eqref{eq:Pidef} and using \eqref{eq:kappaP}
yields $0=\Piscalar-\Piscalar^2$. 
We check \ref{item:introest}: This holds by \ref{item:Gfullest},
note that multiplying and dividing by $x^1$ does not
affect the estimates since $\bar{\Omega}\subset\uHS^n$.
\end{proof}


\section{Application to Corvino-Schoen-type gluing} \label{s13VIII24}

\newcommand{\out}{\textnormal{out}}
\newcommand{\IN}{\textnormal{in}}
\newcommand{\gzero}{\overset{\circ}{g}}
\newcommand{\gin}{g_{\IN}}
\newcommand{\gout}{g_{\out}}
\newcommand{\gsol}{g}

Proposition \ref{prop:main} can be used to prove the following nonlinear result.
We give an informal statement, since the details are routine.
\begin{proposition}[Informal]\label{prop:nonl}
Let $n\ge2$ and $s>n/2$.
Let $\dom\subset \uHS^n$ be an open annulus-type region
with $\bar\dom\subset\uHS^n$.
Let $\gin$ and $\gout$ be two smooth (or $H^s=W^{s,2}$) metrics on $\bar\dom$
that have constant scalar curvature $-n(n-1)$ 
and that are sufficiently close in $H^s$
to the hyperbolic metric $\bb$. 
Then there exists a smooth (or $H^s$) metric $\gsol$ on $\bar\dom$ such that:
\begin{itemize}
\item 
$\RicciScalar(\gsol)=-n(n-1)\mod\ran(\Piscalar)$,
where $\Piscalar$ is defined in Proposition~\ref{prop:main}. 
\item 
$\gsol=\gin$ in a neighborhood of the interior boundary of $\bar\dom$.
\item 
$\gsol=\gout$ in a neighborhood of the exterior boundary of $\bar\dom$.
\end{itemize}
\end{proposition}
\begin{proof}
This is similar to \cite[Proof of Theorem 1.3, Step 1 and 2]{MaoOhTao},
hence we only give a sketch. 
Write $\RicciScalar(\bb+\pert)=-n(n-1)+ \linscalar(\pert) + \nonlin(\pert)$
where $\nonlin(\pert)$ is the nonlinearity, 
given by terms of quadratic and higher order.
Fix an interpolation $\pert$ between $\gin-\bb$ and $\gout-\bb$,
then make the ansatz $\gsol = \bb+\pert+\pert'$
where the correction $\pert'$ has compact support in $\dom$.
It is constructed as the solution of the fixed point  
\begin{equation}\label{eq:fixedpointequation}
\pert' =-\Gintro\big(\linscalar(\pert) + \nonlin(\pert+\pert')\big)\,.
\end{equation}
Then $\gsol$ satisfies all properties stated in the proposition.
\end{proof}

In Proposition \ref{prop:nonl}, the constant scalar curvature equation 
is only solved modulo the range of $\Piscalar$, which has dimension $n+1$.
The full equation $\RicciScalar(g)=-n(n-1)$ holds
if and only if 
the argument of $\Gintro$ in \eqref{eq:fixedpointequation}
is in the kernel of $\Piscalar$, equivalently,
if and only if for all $a=1,\dots,n+1$,
\[ 
\int_{\dom}\kappa_a\big( \linscalar(\pert) + \nonlin(\pert+\pert')\big) \,d\mu_{\bb}=0
\]
where $\kappa_a$ are the static KIDs defined in \eqref{9VII24.51}.
By \eqref{eq:kappaPpert} and the divergence theorem, 
this is equivalent to 
\begin{equation}\label{eq:intautV}
\int_{\p_{\out}\dom} V^{(a)}_i(\gout-\bb) \,dS^i_{\bb}
-
\int_{\p_{\IN}\dom} V^{(a)}_i(\gin-\bb) \,dS^i_{\bb}
= 
- \int_{\dom}\kappa_a  \nonlin(\pert+\pert')\,d\mu_{\bb}
\end{equation}
where the integrals on the left are flux integrals 
over the exterior boundary $\p_{\out}\dom$ and interior
boundary $\p_{\IN}\dom$ respectively, both in the outward direction
(in the low regularity regime one can use averaged flux integrals).

Of particular interest is the case when $\gin$ is given
and one looks for $\gout$ 
in the family of boosted Schwarzschild-AdS metrics, 
which depends on $n+1$ mass and boost parameters $\theta=(\theta_1,\dots,\theta_{n+1})$,
see e.g.~\cite{ChDelayAH}.
The idea is to apply Proposition \ref{prop:nonl}
with $\theta$ as parameters, and to then adjust the parameters
such that the $n+1$ conditions \eqref{eq:intautV} hold.
 
Since the right hand side of \eqref{eq:intautV} is small, 
it is necessary that the flux integrals of $\gin$ are approximately
equal to the flux integrals of some member of the family of
boosted Schwarzschild-AdS metrics.
If this is the case, 
and under appropriate smallness conditions,
one can then construct the exact parameters $\theta$ 
for which \eqref{eq:intautV} holds 
by reformulating \eqref{eq:intautV} as a fixed point equation for $\theta$.  
Arguments of this kind are worked out in detail in \cite{ChDelayAH} in the hyperbolic setting, and in \cite{MaoOhTao,Corvino,CC1.5,ChDelay} in the Euclidean setting. 
 A particularly simple argument (which also shows that the smallness conditions needed can be realised in some cases) can be given for parity-invariant metrics, 
see Example \ref{example:xxxxx} below.

In this manner, Proposition \ref{prop:main} can be used
to glue a constant-scalar-curvature metric $\gin$  
to a boosted Schwarzschild-AdS metric,
in such a way that the glued metric has constant scalar curvature.
This gives a simpler proof of a similar perturbative result in~\cite{ChDelayAH}, under the much weaker differentiability conditions
of Proposition \ref{prop:nonl}, and extends the analysis of~\cite{ChDelayAH} to include dimension $n=2$. 


\begin{example}[Gluing parity-invariant metrics]\label{example:xxxxx}
\label{Ex3IX24.1}
Let $Z:\uHS^n\to\uHS^n$ be the parity isometry of the hyperbolic metric 
about the point $(1,0,\dots,0)$.
In Proposition \ref{prop:nonl}, if the annulus $\Omega$ 
and the metrics $\gin$ and $\gout$ are $Z$-invariant, then 
the glued metric $g$ will be $Z$-invariant, 
by using instead of $\Gintro$ the $Z$-invariant
$\frac12 (\Gintro + Z^* \Gintro Z^*)$.
Then $n$ of the $n+1$ equations in \eqref{eq:intautV} hold
trivially, and \eqref{eq:intautV} becomes a scalar equation.
Then: 
\begin{quote}\textit{If $g_{\IN,\epsilon}$ is a one-parameter family of
$Z$-invariant constant-scalar-curvature metrics depending smoothly on 
a parameter $ \epsilon\in(-1,1)$, and with $g_{\IN,0}=\bb$ the hyperbolic metric, 
then for all $\epsilon$ sufficiently close to zero,
the metric $g_{\IN,\epsilon}$ can be glued,
with constant scalar curvature, 
to a $Z$-invariant Schwarzschild AdS-metric $\gout$.}
\end{quote}
To see this, let $g_{\out,m}$ be the unboosted and $Z$-invariant
Schwarzschild AdS-metric with mass parameter $m$, in particular
$g_{\out,0}=\bb$.
By Proposition \ref{prop:nonl}
one obtains a $Z$-invariant glued metric $g_{m,\epsilon}$ 
that depends parametrically on $m$ and $\epsilon$.
The scalar equation \eqref{eq:intautV} is of the form
$u(m,\epsilon) = 0$, where $u$ is a smooth function 
near the origin that satisfies $u(0,0)=0$. 
Furthermore $\frac{\p u}{\p m}(0,0)\neq0$ since only the 
flux integral over $\p_{\out}\dom$ contributes to this dervative.
Hence by the implicit function theorem, $u(m,\epsilon) = 0$
can be locally solved for $m=m(\epsilon)$, implying the claim.
\hfill$\Box$
\end{example}

\begin{remark}
Non-trivial  families $g_{\IN,\epsilon}$ needed in Example~\ref{Ex3IX24.1}
can be constructed by starting from any smooth $Z$-invariant family of metrics
(not conformal to $\bb$)  passing through $\bb$,
to obtain a constant-scalar-curvature family by conformal rescaling (Yamabe problem \cite{ACF}).
\hfill$\Box$
\end{remark}

 We note that it is also possible to use other families of metrics for $\gout$,
such as the family of metrics with arbitrary energy-momentum
vector constructed in \cite{CortierMass}.


Let us finally mention that given a solution operator 
for the trace-free symmetric divergence equation
with good support and differentiability properties,
then the above discussion may be extended
to the CMC general relativistic constraint equations
(instead of only Schwarzschild-AdS one then has to use Kerr-AdS).
In dimension $n=3$ such an operator is provided in \cite{Andrea},
for $n\ge3$ it will be provided in a forthcoming paper of
Isett, Mao, Oh and Tao.


%
%

\bibliographystyle{amsplain}
\bibliography{OperatorForCorvinoGluingFinal-minimal,AlbaPiotrAndrea-minimal}

\bigskip

  \textsc{University of Vienna}\par\nopagebreak
  \textit{Email:} \protect\url{piotr.chrusciel@univie.ac.at}\par\nopagebreak
  \textit{Homepage:} \protect\url{homepage.univie.ac.at/piotr.chrusciel}\par\nopagebreak
\vskip 2.8mm
  \textsc{University of T\"ubingen}\par\nopagebreak
  \textit{Email:} \protect\url{albachiara.cogo@uni-tuebingen.de}\par\nopagebreak
\vskip 2.8mm
  \textsc{Stanford University}\par\nopagebreak
  \textit{Email:} \protect\url{anuetzi@stanford.edu}

\end{document}